%% file: main2.tex
\documentclass[sigconf]{acmart}
\AtBeginDocument{%
  }

\setcopyright{acmlicensed}
\copyrightyear{2018}
\acmYear{2018}
\acmDOI{XXXXXXX.XXXXXXX}
\acmConference[Conference acronym 'XX]{Make sure to enter the correct
  conference title from your rights confirmation email}{June 03--05,
  2018}{Woodstock, NY}
\acmISBN{978-1-4503-XXXX-X/2018/06}

\input{macros2}

\begin{document}

\title[\method{}: Constrained Decoding without Distorting LLM's Output Intent]{\method{}: Constrained Decoding \\without Distorting LLM's Output Intent}

\author{Yongmin Li}
\email{liyongmin@pku.edu.cn}
\affiliation{%
  \institution{Key Lab of High Confidence Software Technology\\ (Peking University), Ministry of Education;\\ School of Computer Science, Peking University}
  \state{Beijing}
  \country{China}
}

\author{Jia Li}
\email{jia_li@mail.tsinghua.edu.cn}
\affiliation{%
  \institution{College of AI, Tsinghua University}
  \state{Beijing}
  \country{China}
}

\author{Ge Li}
\authornote{Corresponding author.}
\email{lige@pku.edu.cn}
\affiliation{%
  \institution{Key Lab of High Confidence Software Technology\\ (Peking University), Ministry of Education;\\ School of Computer Science, Peking University}
  \state{Beijing}
  \country{China}
}

\author{Zhi Jin}
\email{zhijin@pku.edu.cn}
\affiliation{%
  \institution{Key Lab of High Confidence Software Technology\\ (Peking University), Ministry of Education;\\ School of Computer Science, Peking University}
  \state{Beijing}
  \country{China}
}

\input{chapters2/abstract}

\received{20 February 2007}
\received[revised]{12 March 2009}
\received[accepted]{5 June 2009}

\setcopyright{none} 
\settopmatter{printacmref=false} 
\renewcommand\footnotetextcopyrightpermission[1]{} 

\maketitle

\input{chapters2/introduction}

\input{chapters2/example}

\input{chapters2/background}

\input{chapters2/method}

\input{chapters2/experiments}

\input{chapters2/discussion}

\input{chapters2/related}

\input{chapters2/conclusion}

\begin{acks}
This research is supported by the National Key R\&D Program under Grant No. 2023YFB4503801, the National Natural Science Foundation of China under Grant No. 62192733, 62192730, 62192731, and the Major Program (JD) of Hubei Province (No.2023BAA024).
\end{acks}

\appendix

\bibliographystyle{ACM-Reference-Format}
\bibliography{mybib}

\end{document}

%% file: macros2.tex
\usepackage{mathtools}
\usepackage{amsfonts}
\usepackage{algorithm}
\usepackage{algpseudocode} 
\usepackage{subcaption}
\usepackage{minted} 
\usepackage{pifont}
\usepackage{xcolor}

\newcommand\method{AdapTrack}

\newcommand{\seqof}[3]{{#1}_{#2},\dots,{#1}_{#3}}
\newcommand{\Pof}[1]{P_\text{#1}}
\newcommand{\myright}{\textcolor{green}{\ding{52}}}
\newcommand{\mywrong}{\textcolor{red}{\ding{56}}}
\newcommand\tfapis{TFv1}
\newcommand\tfreal{TFv1Real}
\newcommand\tfapisimprove{360.87\%}
\newcommand\tfrealimprove{38.93\%}
\newcommand\humanevalimprove{7.84\%}
\newcommand\mbppimprove{6.42\%}

%% file: chapters2/abstract.tex
\begin{abstract}

  Language model-based code generation and completion tools 
  have been widely adopted, but they may sometimes produce code that does not meet necessary constraints, such as syntactic correctness or API existence.
  Constrained decoding techniques are developed to help the model generate code that adheres to the constraints 
  by greedily eliminating generation options that violate constraints at each step of the generation process.
  However, there is a severe limitation of constrained decoding, that it distorts the model's output intent, forcing it to produce code that may satisfy the constraint but does not match the development intent and is therefore incorrect.
  In response to this challenge, we propose \method{}.
  By incorporating backtracking into the generation process, \method{} avoids distorting the output intent of the language model, thereby producing results that are not only constraint-compliant but also more semantically aligned with model's output intent.
  On our synthetic API completion dataset, \method{} can achieve up to \textbf{\tfapisimprove{}} improvement compared to constrained decoding;
  on the real-world API completion dataset we collect that exhibits similar issues, \method{} can also achieve an improvement of up to \textbf{\tfrealimprove{}} over constrained decoding;
  in general code genration benchmarks, compared to constrained decoding, \method{} can achieve up to \humanevalimprove{} improvement on HumanEval, and up to \mbppimprove{} improvement on MBPP.
  This indicates that, simply by better adhering to the model's output intent, \method{} can achieve significant improvements.
  We provide a theoretical proof that the distribution produced by \method{} aligns with the model's distribution given the generated tokens, thereby ensuring that the model's output intent is not distorted.
  Experiments on domain-specific language problems show that, compared to existing methods, our approach can provide generation results that are more consistent with the language model's distribution.
\end{abstract}

%% file: chapters2/introduction.tex
\section{Introduction}
\label{sec:introduce}

In recent years, language model-based code generation and completion tools, such as GitHub Copilot \cite{copilot}, have become increasingly popular.
These tools can understand the developer's intent from the context and generate code that aligns with their intent by leveraging the capabilities of the language models \cite{codex,alphacode,gpt4,swebench,codegensurvey}.

Legitimate code must satisfy various constraints, such as adhering to syntax rules and ensuring that the called APIs exist.
Language models are designed solely to predict the distribution of the next token \cite{transformer,gpt}, but lack built-in mechanisms to guarantee that the generated code complies with these constraints.
As a result, the generated code, while potentially aligned with the developer's intent, may sometimes be illegitimate.

Constrained decoding methods \cite{picard,synchromesh,mgd,repilot} are developed to tackle such problems.
Using a constrainer to identify tokens illegal to come after already generated text, it simply eliminates the possibility to generate these illegal tokens by directly setting their probability to zero in the model's output distribution.
Such constrainers are often available \cite{mgd,repilot}, because the Integrated Development Environment (IDE) often provides developers with similar functionality during their development.

\input{floats/example1}

However, there is \textbf{a severe limitation of constrained decoding}\cite{gedi,gad,local-constrain}, that it distorts the output intent of the model and forces the model to produce unnatural code.
In other words, the generated code, while satisfying the constraints, may deviate from the model’s expected functionality and fail to align with the contextual semantics.
Figure \ref{fig:ex1} shows an example to explain this limitation.
The input code context is shown in Figure~\ref{fig:ex1:context}.
The user wants to calculate the rank of a matrix using PyTorch \cite{pytorch} 2.0.
Two APIs can accomplish this task:
\texttt{torch.matrix\_rank} for PyTorch 1.0 to PyTorch 1.8, and  \texttt{torch.linalg.matrix\_rank} for PyTorch after 1.8.
Depending on the user's PyTorch version, the model needs to select the appropriate API.
Figure~\ref{fig:ex1:normal_dist} shows the API names the model intends to output, along with their corresponding probabilities.
Its top candidate is \texttt{matrix\_\-rank}, and second candidate is \texttt{linalg.\-matrix\_\-rank}.
Using a constrainer, it becomes clear that the only legal API in the current version is \texttt{linalg.matrix\_rank}.
Ideally, since we have prohibited the now-illegal \texttt{matrix\_rank} by constrained decoding, the decoding process should most likely output the most possible candidate under constraint, i.e. \texttt{linalg.matrix\_rank}.
However, as shown in Figure~\ref{fig:ex1:constrained_dist}, the API with the highest generation probability in constrained decoding is actually \texttt{matrix\_power} at 57.70\%!
This is entirely unexpected because, without constrained decoding, the generation probability of \texttt{matrix\_power} is only 0.20\%, indicating that the model considers it highly unsuitable for this context.
How does this API suddenly become the top choice under constrained decoding?

Let us examine the constrained decoding process step by step, as illustrated in Figure~\ref{fig:ex1:how}.
In the first time step, the legal options include \texttt{matrix}, \texttt{l}, and several other less possible tokens.
Since there are other valid APIs starting with \texttt{matrix}, such as \texttt{matrix\_power} and \texttt{matrix\_exp}, the constrainer does not eliminate the possibility of generating \texttt{matrix}.
As a result, the model still chooses the token \texttt{matrix} with a probability of 60.46\%.
However, once \texttt{matrix} and the subsequent \texttt{\_} token are generated, we find that the token \texttt{rank}, which accounts for 99.59\% of the model's generation intent, is prohibited by the constrainer.
At this point, since the model has already generated \texttt{matrix\_}, it is forced to continue under constrained decoding, leading it to select the token \texttt{power}, which is entirely misaligned with its output intent.
Consequently, even though the model knows that, aside from \texttt{matrix\_rank}, the most appropriate output here should be \texttt{linalg.matrix\_rank}, constrained decoding compels it to output \texttt{matrix\_power} with a 57.70\% probability -- an answer the model considers highly incorrect.

We observe that, in the process described above, \textbf{the model is ultimately forced to select an API that is entirely misaligned with its generation intent}.
The constrainer can only identify errors one step at a time, so the \texttt{matrix} option remains valid.
However, when we realize that the model's primary intent is thwarted later, we are still compelled to continue generating based on the previously produced tokens, i.e. \texttt{matrix\_}.
As a result, the generated code, while satisfying the constraints, deviates from the model's original output intent.

If a human were writing this code and discovered that the previously used API \texttt{matrix\_rank} was no longer available, but recalled having seen the API \texttt{linalg.matrix\_rank}, they would typically delete the previously written \texttt{matrix\_} and start over again.
Inspired by this behavior, we wonder if a similar mechanism could be introduced into constrained decoding, allowing the model to quickly backtrack to a point where it might have gone wrong and start anew when encountering such difficulties.

Based on this idea, we propose \method{}, which, during the constrained decoding process, dynamically \textbf{backtracks according to the proportion of invalid options among the current next-step choices}.
If this proportion is large, it indicates that the model is unlikely to want to continue generating along this branch and should switch to another branch;
otherwise, it suggests that most of the model's output intent remains on this branch, and we can proceed with generation.
In this way, we address the aforementioned issue by dynamically backtracking to avoid generating results that deviate too far from the model's intent under constraints.

\input{floats/motivating}

We provide proof for our algorithm, demonstrating that the distribution it produces is identical to the model's distribution under constraints given the known generation history.
This proves that our method effectively avoids the issue of constrained decoding distorting the model's generation intent, achieving generation that is both constraint-compliant and aligned with the model's intent.

To demonstrate that \method{} does address the issue we raised in Figure~\ref{fig:ex1} -- using unfamiliar APIs correctly with a constrainer,
we construct a synthetic dataset, \tfapis{}, based on the update in TensorFlow \cite{tensorflow} APIs, deliberately prompting the model to use APIs that are unique to TensorFlow v1.
These v1 APIs are still usable in TensorFlow v2, except that they have a longer prefix.
Since the model is unaware of the TensorFlow version, its output intent inherently includes options for both versions.
By using the constrainer, we examine whether the model can generate these APIs compliant with TensorFlow v1 or v2.
The results show that, across multiple code language models, \method{} significantly improves API generation accuracy compared to naive constrained decoding, with an improvement of up to \textbf{\tfapisimprove{}}.

To ensure that these results are not merely effective on this synthetic dataset, we construct a dataset, \tfreal{}, based on real-world usage examples of these APIs on GitHub, and ask the model to complete the API given the preceding context.
The results on this dataset show that \method{} achieves up to \textbf{\tfrealimprove{}} improvement over constrained decoding, once again demonstrating the effectiveness of \method{}.

To verify this issue's presence and \method{}'s generalizability in constrained decoding of general code generation, we conduct experiments on HumanEval \cite{codex} and MBPP \cite{mbpp} with type constrainer \cite{type-constraint}.
The results show that \method{} can achieve up to \textbf{\humanevalimprove{}} improvement on HumanEval, and up to \textbf{\mbppimprove{}} improvement on MBPP, showing the generalizability of \method{}.

To demonstrate that the generation results of our method are indeed more aligned with the model's output intent, we conduct experiments on domain-specific language problems, sampling 2000 results for each problem consecutively and comparing them with 2000 results by previous approaches.
The results show that the distribution produced by our method has a significantly lower KL divergence from the correct distribution, compared to existing methods, indicating that it better aligns with the actual distribution.

We also conduct robustness analysis on temperature and model size to ensure \method{} works with different hyperparameters.
Due to potential concerns of efficiency and scalability, we conduct experiments with limited backtrack distance, showing that a certain short-range backtracking is enough on \tfapis{} and \tfreal{}.

In summary, our contributions are as follows.
\begin{itemize}
  \item We propose \method{}, a constrained decoding method that can avoid distortion of the model's output intent in greedy constrained decoding.
  \item We prove that \method{}'s distribution aligns with the model's distribution given the generated tokens. Experiments on DSL problems show that the distribution produced by \method{} aligns closer with the model’s true intent.
  \item We construct the \tfapis{} and \tfreal{} datasets to simulate scenarios where the model is unfamiliar with APIs. The results demonstrate that our method significantly improves the accuracy of API generation, with an improvement up to \textbf{\tfapisimprove{}} on \tfapis{} and \textbf{\tfrealimprove{}} on \tfreal{}.
  Experiments on HumanEval and MBPP demonstrate the generalizability of \method{} in general code generation.
\end{itemize}

%% file: floats/example1.tex
\begin{figure*}[t]
\centering
\begin{subfigure}{.3\textwidth}
  \centering

\begin{minted}[escapeinside=!!,mathescape=true]{python}
# pytorch 2.0.0
import torch
def matrix_rank(x: torch.Tensor):
    return torch.!$\big|$!
\end{minted}
  \caption{The example code context. We want an API to calculate the rank of the tensor.}
  \label{fig:ex1:context}
\end{subfigure}\hfill
\begin{subfigure}{.35\textwidth}
  \centering
  \begin{tabular}{l|r|c}
  candidates & probability & usable \\
  \hline
  \verb|matrix_rank| & $57.44\%$ & \mywrong{} \\
  \verb|linalg.matrix_rank| & $19.21\%$ & \myright{} \\
  \verb|...| & & \\
  \verb|matrix_power| & $0.20\%$ & \myright{}\\
  \end{tabular}
  \caption{Candidate APIs the model wants to generate, and whether they are usable in PyTorch 2.0.0.}
  \label{fig:ex1:normal_dist}
\end{subfigure}\hfill
\begin{subfigure}{.3\textwidth}
  \centering
  \begin{tabular}{l|r}
  candidates & probability\\
  \hline
  \verb|matrix_power| & $57.70\%$ \\
  \verb|linalg.matrix_rank| & $20.12\%$ \\
  \verb|...| & \\
  \\
  \end{tabular}
  \caption{Candidate APIs in constrained decoding. Compare with Figure~\ref{fig:ex1:normal_dist}.}
  \label{fig:ex1:constrained_dist}
\end{subfigure}
\par\medskip
\begin{subfigure}{\textwidth}
  \centering
  \includegraphics[width=\linewidth]{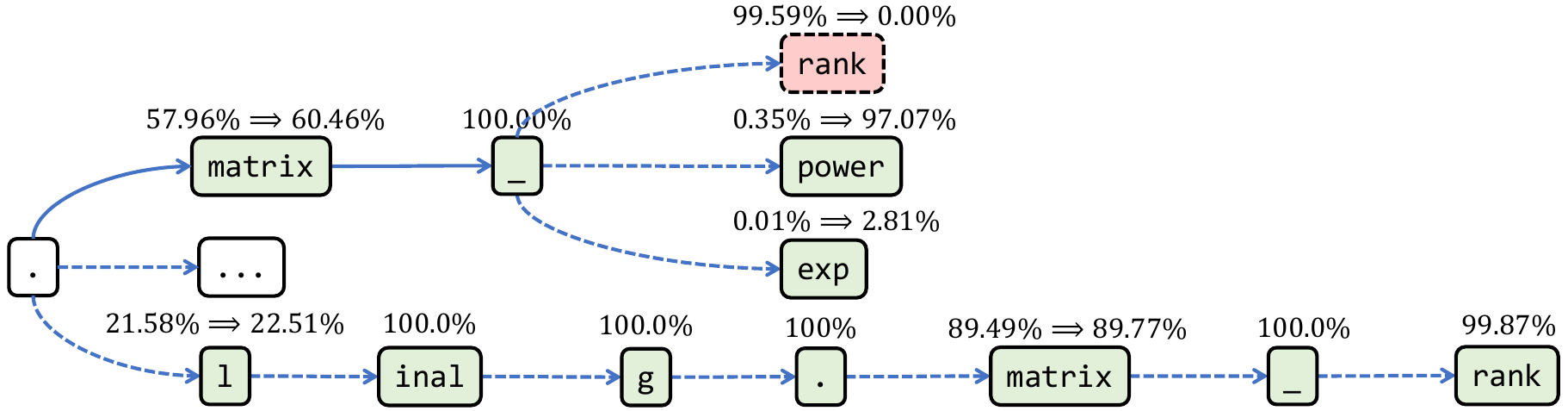}
  \caption{What happened during the constrained decoding.}
  \label{fig:ex1:how}
\end{subfigure}
\caption{An example of incorrect decoding results in constrained decoding.}
\label{fig:ex1}
\end{figure*}

%% file: floats/motivating.tex
\begin{figure*}[h]
\centering
\begin{subfigure}{.3\textwidth}
  \centering
  \includegraphics[width=\linewidth]{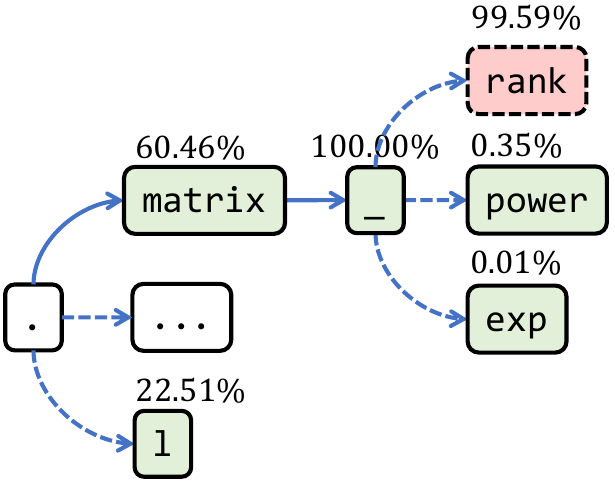}
  \caption{The problematic scene.}
  \label{fig:motiv:1}
\end{subfigure}\hfill
\begin{subfigure}{.31\textwidth}
  \centering
  \includegraphics[width=\linewidth]{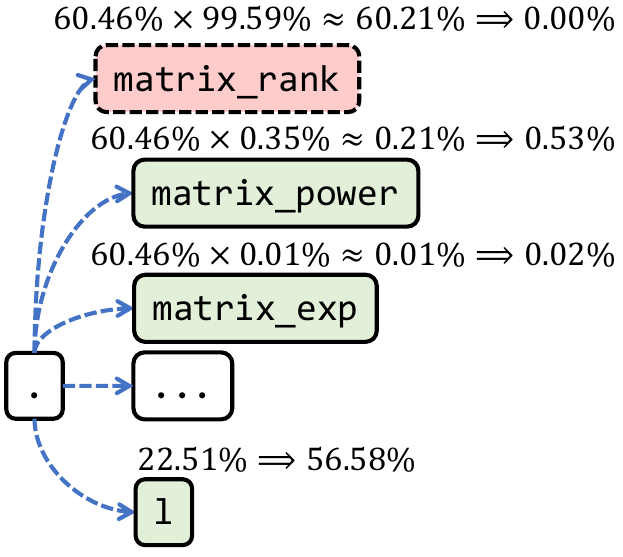}
  \caption{Figure~\ref{fig:motiv:1} is actually equivalent to this.}
  \label{fig:motiv:2}
\end{subfigure}\hfill
\begin{subfigure}{.35\textwidth}
  \centering
  \includegraphics[width=\linewidth]{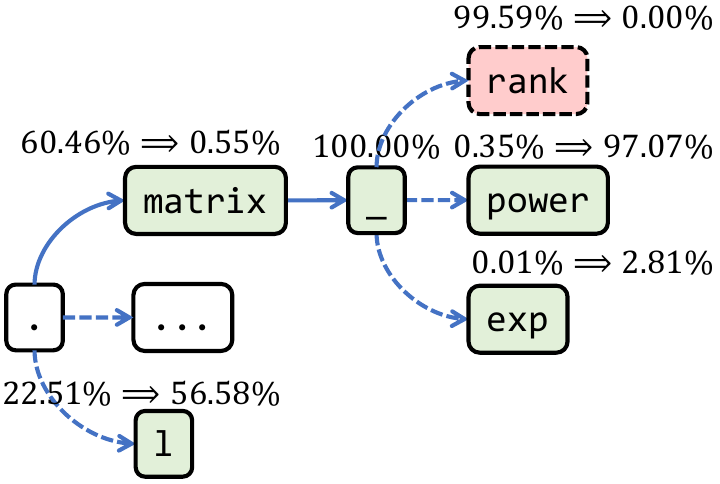}
  \caption{Collect the probabilities in Figure~\ref{fig:motiv:2} gives this.}
  \label{fig:motiv:3}
\end{subfigure}

\begin{subfigure}{.32\textwidth}
  \centering
  \includegraphics[width=\linewidth]{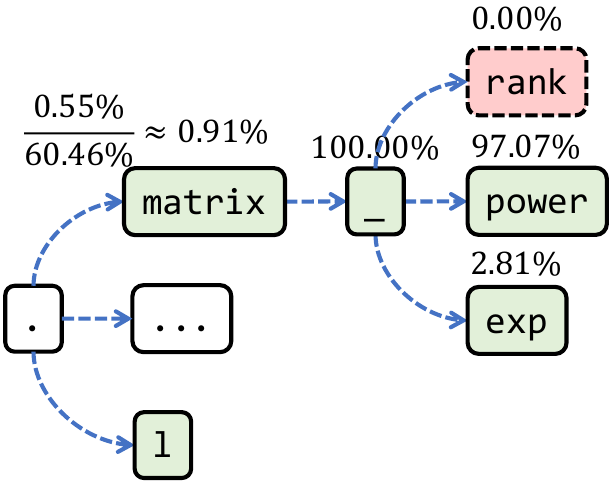}
  \caption{To actually choose the samples the way in Figure~\ref{fig:motiv:3}, we have to do rejection sampling.}
  \label{fig:motiv:4}
\end{subfigure}\quad\quad
\begin{subfigure}{.32\textwidth}
  \centering
  \includegraphics[width=\linewidth]{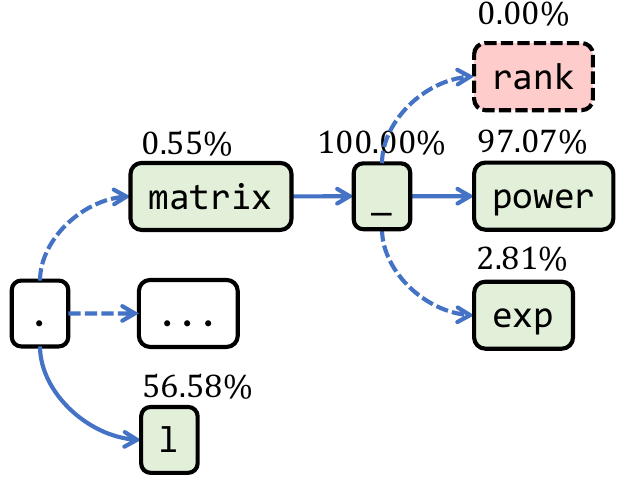}
  \caption{After rejection sampling, the generation probably looks like what we want.}
  \label{fig:motiv:5}
\end{subfigure}%
\caption{Motivating example: how backtracking works on the example in Figure~\ref{fig:ex1}.}
\label{fig:motiv}
\end{figure*}

%% file: chapters2/example.tex
\section{Motivating example}
\label{sec:example}

We have seen the problem of constrained decoding in Figure~\ref{fig:ex1}.
The problem occurs when the decoding process eliminates the probability of one API, but fails to generate another highly possible API with a different prefix.
We propose to introduce backtracking into the decoding process as a potential solution,
and in this section we will delve into a more comprehensive discussion of this approach.

Figure~\ref{fig:motiv} shows step by step how to utilize this intuition to modify the constrained decoding process.
The arrows between tokens represent the generation order.
A solid arrow indicates that the current state of generation includes this step, whereas a dashed arrow signifies other potential generation options.
A red dashed box denotes tokens that are invalid in the current context, and a green solid box represents valid tokens.
The numbers on the tokens indicate their conditional generation probabilities given the preceding context.

Figure~\ref{fig:motiv:1} illustrates the scenario where the issue mentioned in Figure~\ref{fig:ex1} occurs.
At this point, we have generated \texttt{matrix\_}, and subsequently identified the following token options, their validity, and how much the language model wants to generate them.
We observe that the model has a 99.59\% probability to generate the \textbf{invalid} \texttt{rank} token.

We now consider the overall probability of the sequence from the first token up to each candidate option, as shown in Figure~\ref{fig:motiv:2}.
For example, the probability of the model generating \texttt{matrix\_rank} is roughly 60.21\%.
After eliminating invalid options, we can renormalize the probabilities of all remaining valid options.
This process ensures that all valid options experience an equivalent proportional increase in their probabilities.

We can reorganize the probabilities, resulting in Figure~\ref{fig:motiv:3}.
The probability of \texttt{matrix} has decreased from 60.46\% to 0.55\%, which aligns more closely with our intuition that, after learning that the top candidate \texttt{matrix\_rank} is invalid, the model should aim to generate the second candidate, which does not begin with \texttt{matrix}.
This happens because we have performed a global probability adjustment after eliminating the invalid option \texttt{matrix\_rank}, thereby allowing the probability of \texttt{matrix} to correctly decrease.

Now we want to \textbf{backtrack and change the initial choice} of \texttt{matrix}, so that we can follow the output intent of the model to not generate \texttt{matrix} and circumvent the limitation of constrained decoding.
But now comes the question: given that we have already generated \texttt{matrix} with a probability of 60.46\%, how can we adjust this probability to 0.55\%?
One option is to directly resample.
In practice, however, this would imply that we would resample the entire sequence regardless of how minor the new error identified by the constrainer is, which would be wasteful of time and computational resources.
We have an intuition that if the probability change is substantial, we are likely inclined to depart from this branch; otherwise, we probably wish to remain within the same branch.

To formalize this intuition, we propose the use of rejection sampling for \textbf{adaptive backtracking}, as shown in Figure~\ref{fig:motiv:4}.
We have generated \texttt{matrix} with a probability of 60.46\%.
We can perform a second sampling with an acceptance rate of roughly 0.91\%.
If the sampling succeeds, we can continue to use \texttt{matrix}; otherwise, we switch to another token.
Thus, by combining the results of the two sampling processes, we achieve a sampling outcome where \texttt{matrix} is adopted with a probability of 0.55\%.

After this resampling, one possible result is shown in Figure~\ref{fig:motiv:5}, where the first token is changed to \texttt{l}.
This directs us towards the generation of the second candidate, \texttt{linalg\-.matrix\_rank}.

%% file: chapters2/background.tex
\section{Background}
\label{sec:background}

\subsection{Goal}
We have demonstrated, by example, that the output of constrained decoding can be far away from the model’s intent, which can be shown by the large difference between the probability distribution of constrained decoding and that of the language model.

To keep the output aligned with the model’s intent, the best option is to perfectly align with the probability distribution of the language model while keeping the constraints.
In this subsection, we will formalize this goal.

Given a context $s_{<i}$, an autoregressive language model can produce a distribution over the next token $P(s_i|s_{<i})$.
This allows us to define a distribution over all possible sequences, i.e., the joint probability distribution of the entire sequence,
$$
P(\seqof{s}{0}{n}) = \prod_{i=0}^n P(s_i|s_{<i}).
$$

A constrainer is a function defined on a sequence $c: s \mapsto \{0, 1\}$, where the result is $0$ if all sequences that start with this sequence as their prefix is invalid, and $1$ otherwise.
We use $C(s) = \{t | c(s :: [t]) = 1\}$ to represent all valid tokens after the prefix $s$.

Our task is to define a distribution $P_c$ over all valid sequences, such that the probability of generating any sequence is proportional to its probability under the language model. This can be formally expressed as follows,
$$
P(s | c) \propto P(s) c(s), \forall s \text{ that is complete}.
$$

Or equivalently, for any two valid complete sequences $s$ and $s'$, we should have the following property,
\begin{equation}
\frac{P(s | c)}{P(s' | c)} = \frac{P(s)}{P(s')}. \label{eq:fair}
\end{equation}

\subsection{Constrained decoding and its problem}

Constrained decoding is an ad-hoc modification of rejection sampling, where illegal options are directly removed during the sampling process.
Specifically, instead of repeatedly sampling and discarding invalid sequences, constrained decoding eliminates invalid options at each step, ensuring that only valid tokens are considered for generation.

Consider the generation probability of a valid sequence $ s = (\seqof{s}{0}{n}) $ under constrained decoding. At each step of sampling \( s_i \), the probability can be viewed as a form of rejection sampling, where the probability of sampling \( s_i \) given the prefix \( s_{<i} \) is normalized over all valid continuations. Specifically, the probability of sampling \( s_i \) is given by the following equation,
\begin{equation*}
\begin{split}
\Pof{con}(s_i \mid s_{<i}) &\propto P(s_i \mid s_{<i}) \lBrack s_i \in C(s_{<i}) \rBrack, \\
\Pof{con}(s_i \mid s_{<i}) &= \frac{P(s_i \mid s_{<i})\lBrack s_i \in C(s_{<i}) \rBrack}{\sum_{t \in C(s_{<i})} P(t \mid s_{<i})}.
\end{split}
\end{equation*}

Let us compare the generation probabilities of two valid sequences \( s \) and \( s' \).
\[
\frac{\Pof{con}(s)}{\Pof{con}(s')} = \frac{P(s)}{P(s')} \frac{\prod_{i=0}^{n'} \sum_{t \in C(s'_{<i})} P(t \mid s'_{<i})}{\prod_{i=0}^n \sum_{t \in C(s_{<i})} P(t \mid s_{<i})}.
\]

When compared to the equation~\ref{eq:fair}, the second term in the product, which depends on the sums over \( C(s_{<i}) \) and \( C(s'_{<i}) \), cannot be canceled out. This shows that the generation probabilities of sequences under constrained decoding are distorted relative to their original language model probabilities.

%% file: chapters2/method.tex
\section{\method}
\label{sec:method}

After defining our goal and necessary notations in section~\ref{sec:background}, we can now describe our algorithm \method{}.
Overall, our algorithm keeps track of the validity of all known incomplete sequences and maintains a sample for each prefix based on the validity information.
Our algorithm repeatedly samples an incomplete sequence and collects the validity information into the sample.
We compare the validity information before and after this sampling and perform rejection sampling to correct the samples for each prefix.
When the rejection happens, our algorithm will next time select a different token on this prefix.
This is equivalent to \textbf{backtracking}, because we no longer use the same partially generation sequence.

Let us start from the beginning.
We want to sample the distribution $P(s|c)$, which can be decomposed as a series of sampling from $P(s_i|c, s_{<i})$ as follows,
$$
P(s|c) = \prod_i P(s_i|c, s_{<i}).
$$

$P(s_i|c, s_{<i})$ can be decomposed by the Bayesian rule as follows,
$$
P(s_i|c, s_{<i}) \propto P(s_i | s_{<i}) P(c|s_{<i} :: [s_{i}]).
$$

For every prefix $x$, we define our current estimation of $P(c | x)$ as $Q[x]$, where $Q$ is an associative array.
We call a prefix $x$ is \textbf{new} if we have not calculated $P(\cdot | x)$; otherwise we call it \textbf{old}.
If the prefix $x$ is invalid, $Q[x]$ should be $0$; otherwise, if $x$ is old, which means we know possible valid sequences starting from $x$, we simply aggregate the information; otherwise, we just assume that its children are always valid and use $1$ as an estimation.
\begin{equation}\label{eq:Q}
Q[x] := \begin{cases}
    c(x) & \text{if } c(x) = 0 \text{ or } x \text{ is new};\\
    \sum_u P(u | x) Q[x :: [u]] & \text{if } c(x) = 1 \text{ and } x \text{ is old}.
\end{cases}
\end{equation}
We use $P_Q(\cdot)$ to denote the distribution of $P(\cdot | c)$ estimated with $Q$,
\begin{equation}\label{eq:P_Q}
    P_Q(s_i | s_{<i}) \propto P(s_i | s_{<i}) Q[s_{<i} :: [s_i]].
\end{equation}

For every prefix $x$ we have encountered during decoding, we maintain a next token $N[x]$ sampled from $P_Q(\cdot|x)$, where $N$ is an associative array.
We can concatenate all the next token samples to get a sample from $P_Q(\cdot)$, as shown in algorithm~\ref{alg:collect_sequence}.

\input{floats/collect-sequence}

Our decoding process first samples a prefix $s$ from $P_Q(\cdot)$ as indicated by $N$.
If it is a complete sequence, we can just return it.
Otherwise, we call the constrainer for every possible next token and record their information in $Q$ accordingly.
We sample $N[s]$ from all the valid next tokens by their probability from the language model (the same distribution as $P_Q(\cdot|s)$).
We update $Q[x]$ for all prefixes $x$ of the sequence $s$ and decide \textbf{whether we want to backtrack} at this prefix by updating the prefix's next token $N[x]$.

The estimated validity $Q[x :: [N[x]]]$ of the current next token $N[x]$ decreases during the update because we always overestimate the validity.
This reduces the probability of drawing $N[x]$ from $P_Q(\cdot | x)$, as shown in formula~\ref{eq:P_Q}.
We decide whether we continue to use $N[x]$ as the sample from the updated $P_Q(\cdot | x)$ by rejection sampling.
Denote $Q$ before updating as $Q$, and that after updating as $Q'$.
We sample a Bernoulli distribution with the probability of $P_{Q'}(N[x]|x) / P_Q(N[x] | x)$ to keep the current next token $N[x]$ unchanged, and otherwise change to other tokens \footnote{This probability is different from $Q'[x :: [N[x]]] / Q[x :: [N[x]]]$, because the denominator in $P_Q(\cdot | x)$ is also changed.}.
If we change to another next token, which means \textbf{backtracking occurs}, we sample a new next token from all \textbf{other} tokens with probability $P_{Q}(\cdot | c, x)$, but we explicitly exclude the old $n:=N[x]$.
\begin{equation}\label{eq:P'}    
P'_Q(n' | x, n) \propto P_Q(n' | x) \lBrack n' \neq n \rBrack.
\end{equation}

Collectively, our procedure can be summarized in algorithm~\ref{alg:mine}.

\input{floats/main-alg}

\section{The proof of \method{}}
\label{sec:proof}

\method{} focuses on not distorting the model distribution under constraint and, therefore, not distorting its output intent.
We now give a proof of this statement.

\begin{theorem}
    Algorithm~\ref{alg:mine} only returns a valid complete sequence, and the probability of generating different sequences that have been explored during decoding satisfies equation~\ref{eq:fair}.
\end{theorem}

\begin{proof}
    The only \textbf{return} statement is step~\ref{step:return}, so the returned sequence is always complete.
    Also, the returned sequence is collected from $N$, where $N[x] \sim P_Q(\cdot | x) \propto P(\cdot | x) Q[x :: [\cdot]]$ always guarantees that $N[x]$ is a valid subsequent token for any prefix $x$, so the algorithm always returns valid well-ended sequences.

    Collecting one sample sequence from $N$ is equivalent to sampling a result from $P_Q(\cdot)$ because we always maintain this property in the algorithm by rejection sampling, as shown in theorem~\ref{thm:updateN}.

    For any new sequence $s = (\seqof{s}{0}{n})$,
    \begin{equation*}
    \begin{split}
        P_Q(s) &= \prod_{i=0}^n P_Q(s_i | s_{<i})
        = \prod_{i=0}^n \frac{P(s_i | s_{<i}) Q[s_{<i} :: [s_i]]}{\sum_u P(u | s_{<i}) Q[s_{<i} :: [u]]}, \\
        &= \prod_{t=0}^n \frac{P(s_i | s_{<i}) Q[s_{\leq i}]}{Q[s_{<i}]}
        = \frac{Q[s]}{Q[\varepsilon]} P(s)
        = \frac{1}{Q[\varepsilon]} P(s). \qedhere
    \end{split}
    \end{equation*}
    
\end{proof}

\begin{theorem}\label{thm:updateN}
    The update process of $N$ in steps~\ref{step:rej:start} to \ref{step:rej:end} in algorithm~\ref{alg:mine} makes $N[x]$ follows distribution $P_Q(\cdot | x)$ for prefix $x$.
\end{theorem}
\begin{proof}
    We give a proof by induction on the update of $N[x]$.
    
    \paragraph{Base case}
    In step~\ref{step:N:first_draw}, $N[x]$ is drawn from $P_Q(\cdot | x)$.

    \paragraph{Induction step}
    When the distribution $Q[x]$ is updated in step~\ref{step:updateQ}, 
    denote the original $N[x]$ by $n$.

    By rejection sampling with an acceptance rate of $P_{Q'}(n | x) / P_Q(n | x) \leq 1$, the total probability of generating $n$ as a sample after the prefix $x$ is as follows.
    $$
    P_{Q}(n | x) \times \frac{P_{Q'}(n | x)}{P_Q(n | x)} = P_{Q'}(n | x).
    $$

    The probability of generating another token $n' \neq n$ as a sample after the prefix $x$ is as follows.
    \begin{equation*}
    \begin{split}
        &P_Q(n' | x) + P_Q(n | x) \times \left(1 -  \frac{P_{Q'}(n | x)}{P_Q(n | x)}\right)\times \frac{P_Q(n' | x)}{1 - P_Q(n | x)}, \\
        = &(1 - P_{Q'}(n|x)) \times \frac{P_Q(n' | x)}{1 - P_Q(n | x)},\\
        = &(1 - P_{Q'}(n|x)) \times \frac{P(n'|x) Q[x :: [n']]}{\sum_{u \neq n} P(u|x) Q[x :: [u]]},\\
        = &(1 - P_{Q'}(n|x)) \times \frac{P(n'|x) Q'[x :: [n'])]}{\sum_{u \neq n} P_\text{LM}(u|x) Q'[x :: [u]]},\\
        = &(1 - P_{Q'}(n|x)) \times \frac{P_{Q'}(n' | x)}{1 - P_{Q'}(n | x)}
        = P_{Q'}(n' | x). \qedhere
    \end{split}
    \end{equation*}
\end{proof}

%% file: floats/collect-sequence.tex
\begin{algorithm}
\caption{Collect a sequence from next tokens}
\label{alg:collect_sequence}
\begin{algorithmic}[1]

\Procedure{CollectSequence}{$N$}
    \State $s \gets \varepsilon$
    \While{$\textsc{Contains}(N, s)$}
        \State $s \gets s :: [N[s]]$
    \EndWhile
    \State \textbf{return} $s$
\EndProcedure
\end{algorithmic}
\end{algorithm}

%% file: floats/main-alg.tex
\begin{algorithm}
\caption{\method{}}
\label{alg:mine}
\begin{algorithmic}[1]

\Procedure{BacktrackSample}{$\Pof{LM}, c$}

\State $Q \gets \varnothing$ \Comment{Estimate of $P(c | s)$ for all known $s$}
\State $N \gets \varnothing$ \Comment{Sample of $P_Q(\cdot | s)$ for all known $s$}

\While{\textbf{true}}
    \State $s \gets \textsc{CollectSequence}(N)$
    \If{$s_{-1}$ is complete}
        \State \textbf{return} $s$ \label{step:return}
    \EndIf
    \State $N[s] \sim P_Q(\cdot | s)$ \label{step:N:first_draw}
    \State Update $Q$ by formula~\ref{eq:Q}, denote the updated $Q$ as $Q'$\label{step:updateQ} \par\Comment{Update from the longest sequence to the shortest}
    \For{each prefix $x$ of $s$ excluding $s$ itself}\label{step:rej:start}
        \State $n \gets N[x]$ \Comment{The old $N[x]$}
        \State $p \gets P_Q(n | x)$
        \State $p' \gets P_{Q'}(n | x)$
        \State $b \sim \text{Bernoulli}(p' / p)$ \Comment{Rejection sample} \label{step:rej}
        \If{$b = 0$} \Comment{Backtrack}
            \State $N[x] \sim P'_Q(\cdot | x, n)$ as formula~\ref{eq:P'}
        \EndIf
    \EndFor\label{step:rej:end}
\EndWhile

\EndProcedure
\end{algorithmic}
\end{algorithm}

%% file: chapters2/experiments.tex
\section{Evaluation}
\label{sec:experiments}

We have proven \method{} aligns with the language model in section~\ref{sec:proof}.
In this section, we evaluate \method{} against other decoding methods to answer the following research questions:

\paragraph{RQ1} How does \method{} perform in predicting correct and valid APIs?

\paragraph{RQ2} How does \method{} perform in predicting correct and valid APIs in real-world scenarios?

\paragraph{RQ3} How does \method{} perform in general code generation?

\paragraph{RQ4} How does \method{} align with the oracle distribution of language models under constraint?

\subsection{RQ1: Performance of \method{} in predicting correct and valid APIs}

\input{floats/tfv1-results}

As the library evolves, it often happens that an old API is transitioned into another new API.
If the user has updated their library, only the new APIs should be used; if not, only the old APIs should be used.
In RQ1, we aim to simulate this scenario.
We ask the model to use the same API with the old and the new versions of the same library, respectively.
The APIs used by the model should not only \textbf{have the correct semantics}, but also \textbf{be valid in the corresponding library version}.
With the new library, the model should adopt the new API; with the old one, it should stick to the old API.

\input{floats/tfv1-example}

\paragraph{Dataset}
We construct a dataset, \tfapis{}, based on the update from TensorFlow v1 to v2.
We prompt the model to generate an API unique to v1 in a context where it cannot infer the library version and then evaluate its ability to correctly generate the API \textbf{in both v1 and v2 settings}.
The version information is given to the constrainer, so the constrainer is aware of all the valid APIs in the setting.
To guide the model toward using a specific API, we explicitly annotate the function as a wrapper, encouraging the model to directly forward calls to this API.
An example is illustrated in Figure~\ref{fig:tf-old:ex}.
For this example, the correct answer is \texttt{.losses.\-mean\_\-pairwise\_\-squared\_\-error} and \texttt{.compat.\-v1.\-losses.\-mean\_\-pairwise\_\-squared\_\-error} for Tensor\-Flow v1 and v2 settings, respectively.

We collect all legal APIs under TensorFlow 2.16 from the official TensorFlow website.
From them, we filter all APIs starting with \texttt{tf.\-compat.\-v1}.
We exclude APIs that have counterparts in v2 based on the last part of the API name.
To avoid interference from classes and methods, we retain only APIs where all parts of its name start with a lowercase letter.
To avoid interference from subpackages, we remove all APIs that can be a prefix of another API.
Ultimately, we obtained 419 standalone TensorFlow v1-specific APIs.

\paragraph{Metrics}
We use exact match (EM) metrics.
We report EM@$k$ \cite{codex}, 
\begin{equation}
\text{EM@}k := \mathbb{E}_{\text{Problems}}\left[1 - \frac{\binom{n - c}{k}}{\binom{n}{k}}\right], \label{eq:em_at_k}
\end{equation}
where $n$ is the sample size, $c$ is how many samples exactly match the oracle in $n$ samples, and $k$ is in 1, 3, 5, 10, 20.
EM@$k$ reflects the expected value of achieving an exact match within $k$ attempts.

For each method on each problem in each setting, we sample one result by greedy decoding and then sample 20 results randomly.
For EM@1, we use the greedy result, and the sample size is 1.
For EM@$k$ where $k > 1$, we calculate the metrics from the 20 samples.
Because unconstrained decoding does not depend on the version information, we reuse the same sampling results and only change the oracle between the two settings.

\paragraph{Baselines}
We compare with two baselines, unconstrained decoding and constrained decoding.
Unconstrained decoding is decoding without a constrainer.
At each time step, the next token is sampled from the model distribution without filtering.
Constrained decoding is the decoding method we describe in section~\ref{sec:background}.
At each time step, the probabilities of invalid tokens are set to 0 in the model distribution before sampling.

\paragraph{Implementation}
For the TensorFlow v2 setting, we use all collected TensorFlow APIs as valid APIs in the constrainer.
For TensorFlow v1 setting, we use all the v1 APIs as valid APIs in the constrainer, but we remove their additional prefix \texttt{compat\.v1}, to keep the same syntax as in TensorFlow v1.
We implement the \method{} algorithm with PyTorch\cite{pytorch} and Transformers\cite{transformers} library.
We conduct experiments on a 64-core Intel Xeon machine with 8 NVIDIA RTX A6000 GPUs.

\paragraph{Models}
We conduct experiments on Qwen2.5 Coder 7B\cite{qwen2.5coder}, Deep\-Seek Coder Base 6.7B\cite{deepseekcoder}, Star\-Coder2 7B\cite{starcoder2} and Code\-Llama Python 7B\cite{codellama}.
We only use base models in this experiment, because the task of code completion aligns better with the training target of the base models than that of the instruct models.

\paragraph{Results}
As shown in Table~\ref{tab:tf1-apis},
\textbf{for both TensorFlow v1 and v2 settings, \method{} performs significantly better than constrained decoding, with an improvement up to \tfapisimprove{}.}

In the TensorFlow v1 setting, Qwen2.5 Coder 7B can already achieve an EM@1 of 61.58\%.
If we equip it with a constrainer, the result can be improved to 66.11\%, indicating the model may still generate some API calls invalid in TensorFlow v1 and eliminated by the constrainer.
But using \method{}, the result can be further improved to 80.19\%, indicating that, in roughly 14\% cases, the model considers the oracle as its best candidate under constraint but fails to generate it, because constrained decoding does not allow the model to modify already generated tokens.
Since \method{} does allow the model to modify them to better express its output intent, \method{} performs significantly better than simple constrained decoding.

It is more dramatic in TensorFlow v2.
Without a constrainer, Qwen2.5 Coder 7B can only achieve an EM@1 of 7.16\%.
This result should be considered together with EM@1 in the TensorFlow v1 environment.
For 61.58\% cases, the model wants to generate v1-compliant APIs as the best candidate, and then it cannot generate v2-compliant APIs by unconstrained decoding.
If we add a constrainer, this result can be improved to 18.38\%, because APIs that are invalid in TensorFlow v2 cannot be generated now.

Now consider the results of \method{}.
With the help of the constrainer to remove invalid APIs, the model can now freely generate the best remaining option in its candidates, resulting in an EM@1 of 54. 89\%, which is an improvement of \textbf{198. 70\%} over the 18. 38\% EM@1 of constrained decoding.
This is even better than EM@10 of constrained decoding, further showing the importance of not distorting the language model’s output intent.

Similar analysis can be applied to the results of other models, while the biggest improvement, \textbf{\tfapisimprove{}} compared to constrained decoding, happens at EM@1 of DeepSeek Coder Base 6.7B in TensorFlow v2 environment, from 10.98\% to 50.60\%.

\subsection{RQ2: Performance of \method{} in real-world API completion}

\input{floats/tfv1real-results}

We have found that \method{} is extremely effective on \texttt{tf.\-compat.\-v1.} APIs in RQ1, but the results come from a synthetic dataset, \tfapis{}.
We want to know whether this phenomenon also occurs in real-world code completion.

\paragraph{Dataset}
We collect real usage of such APIs from GitHub.
For every v1 API in \tfapis{}, we search for \texttt{tf.\-compat.\-v1.\-API(}, and collect 46,785 Python files from GitHub.
We deduplicate these files and split each file on the first usage of \texttt{tf.\-compat.\-v1.\-} in code.
We then remove all files that contain \texttt{tf.\-compat.\-v1} in the prefix.
If the suffix does not start with an API call, the file is also removed.
After these cleaning steps, we now obtain 14,237 files.
We use the prefix as the completion prefix and the starting API call in the suffix as the oracle.

We then calculate the length distribution of the prefixes.
We use the tokenizer of CodeLlama 7B Python to tokenize the prefixes because it has a relatively small vocabulary of 32,000 tokens, which makes it easier to overflow the model’s max position after tokenization.
The average length of the prefixes is 1,624.78 tokens, and 90\% of the prefixes have less than or equal to 3,463.0 tokens.
To make the dataset friendlier to models with smaller context windows, we remove files whose prefix is longer than 4,096 - 512 = 3,584 tokens.
This gives us 12,883 files.
We then randomly pick 1,000 files as the \tfreal{} dataset.

\paragraph{Results}
As shown in Table~\ref{tab:tf1-apis-real},
\textbf{\method{} performs significantly better than constrained decoding for all models, with an improvement up to \tfrealimprove{}.}

Take Qwen2.5 Coder 7B as an example.
Qwen2.5 Coder 7B can already correctly generate 21.70\%.
With a constrainer to eliminate invalid options, the EM@1 can be further improved to 35.60\%, indicating that the model is generating many invalid APIs when not constrained.
With \method{}, this result can be further improved to 49.30\%, an improvement of 38.48\% compared to the result of constrained decoding, indicating that in constrained decoding, the model is forced to generate many suboptimal options that are misaligned with the model’s output intent.
The improvement is not as significant as in the \tfapis{} dataset. 
Nevertheless, 49.30\% is still a large improvement.

Similar analysis can be applied to the results of other models, while the largest improvement, \textbf{\tfrealimprove{}} compared to constrained decoding, happens in the results of StarCoder2 7B at EM@1, from 29.80\% to 41.40\%.

\subsection{RQ3: Performance of \method{} in general code generation}

\input{floats/codegen-results}

In RQ1 and RQ2, we conduct experiments on two API completion datasets, \tfapis{} and \tfreal{}.
Now we want to know whether \method{} is effective in general code generation.

\paragraph{Dataset}
We adopt the widely used HumanEval \cite{codex} and MBPP \cite{mbpp} benchmarks.
We use a type constrainer \cite{type-constraint} to ensure the generated programs are well-typed.
Because the constrainer is only available for programs written in TypeScript, we use the TypeScript-translated versions from the MultiPL-E dataset \cite{multipl-e}.

\paragraph{Metrics}
We report pass@$k$ \cite{codex}, which is given by replaces $c$ in equation~\ref{eq:em_at_k} with how many samples can pass the test.

\paragraph{Sampling strategy}
It is too time-consuming to check the validity of the whole vocabulary after every prefix.
We introduce a variant of top-$p$ (nucleus) sampling \cite{top-p} to reduce the tokens to be checked.
We check the vocabulary from the most possible token to the least possible one.
After calculating the validity of each token, we calculate the total probability of known valid tokens, and its ratio to the total probability of possibly valid tokens (i.e. known valid tokens + unchecked tokens).
Once this ratio exceeds $p$, all remaining unchecked tokens are considered invalid.
In unconstrained decoding, this sampling strategy naturally becomes the classic top-$p$ sampling, because the constrainer now considers every token to be valid.

\paragraph{Implementation}
We modify the constrainer to fix several bugs and speed up the calculation.
We open-source the modified package together with our source code.

We restart generation when the constrainer cannot determine validity of one token in 60 seconds.
Since greedy sampling always returns the same answer for the same problem, if greedy sampling is restarted 20 times for the same problem, the sample is considered empty.

\paragraph{Results}
As shown in Table~\ref{tab:codegen}, \textbf{\method{} almost always outperforms constrained decoding, with an improvement of pass@1 up to \humanevalimprove{} on HumanEval and \mbppimprove{} on MBPP, compared to constrained decoding},
indicating that \textbf{distortion also widely happens to constrained decoding in general code generation}, which can be fixed by \method{} to improve the performance.
The only outlier is pass@20 of CodeLlama 7B on HumanEval.

When compared to unconstrained decoding, \method{} sometimes performs worse (pass@1 of StarCoder2 7B and CodeLlama 7B on MBPP).
This is due to \textbf{the constrainer's limitation}. \textbf{}
Note that pass@1 is calculated from the greedy samples.
If the greedy samples of uncontrained decoding can pass the test, it should pass the constrainer and be greedily sampled by constrained decoding.
However, in these cases, unconstrained decoding outperforms constrained one.
Nevertheless, \method{} can still improve the performance of constrained decoding by fixing the distribution distortion.

\subsection{RQ4: Alignment between \method{} and the oracle distribution}

We have conducted experiments on both API completion and code generation.
Now we want to verify the statement that the distribution produced by \method{} aligns better with the oracle distribution, as an explanation of the previous results.
Some previous work \cite{gad} also focuses on fixing the language model’s probability in constrained decoding, so we follow their work and duplicate their experiments on \method{} in RQ4.

\paragraph{Datasets}

Following previous work \cite{gad}, we use 4 datasets in domain-specific languages.

The SLIA (strings with linear integer arithmetic) dataset and the INV-BV (loop invariant generation with bit-vector arithmetic) dataset each contain 15 problems from the Syntax-Guided Synthesis (SyGuS) problems \cite{sygus}.
Given the grammar and specification of a domain-specific language, the goal is to generate a function that satisfies the specification using the given grammar.
Each problem comes with its own grammar and logical specification, provided to the model in the prompt.
Each prompt also contains 3 in-context examples in the form of (Problem, Solution) pairs.

The CP (constituency parsing) dataset contains 6 problems of constituency parsing\cite{cp}, i.e., to generate the constituency parse trees from English sentences.
The constrains are used to make sure that the parentheses in the generated sequence are well-paired.

The binary dataset \cite{gad} is specifically designed to demonstrate this problem.
There is only one problem in this dataset that asks the model to generate a random sequence of 5 bits, but the constraint only allows \texttt{00000} and any 5-bit string that starts with \texttt{1}.
For a correctly constrained decoding method, the probability to generate \texttt{00000} should be roughly $1/17$, because there are $17$ valid samples, and they should all have the same probability.
But constrained decoding will give \texttt{00000} a much higher probability, roughly $1/2$, because in the first step, there are only 2 valid options, i.e., \texttt{0} and \texttt{1}, and their probabilities are roughly the same.

\paragraph{Metrics}
Following previous work\cite{gad}, for each problem in the SLIA, INV-BV and CP datasets, we generate 2000 samples; for the binary dataset, we generate 100 samples.
We use as metrics the Kullback–Leibler (KL) divergence between the distribution formed by the samples and the language model distribution.
\textbf{Lower the KL divergence, the closer the two distributions.}
For any distribution $Q$ that only contains constrained samples, the KL divergence between $Q$ and $\Pof{LM}$ reflects the KL divergence between $Q$ and $P(\cdot | c)$.
\begin{multline*}
    D_{\text{KL}}(Q \| P(\cdot | c)) = \sum_s Q(s) \log \frac{Q(s)}{P(s | c)}
    = \sum_s Q(s) \log \frac{Q(s)}{\Pof{LM}(s) Z}, \\
    = \sum_s Q(s) \log \frac{Q(s)}{\Pof{LM}(s)} - \log Z
    = D_{\text{KL}}(Q \| P_{\text{LM}}) - \log Z,
\end{multline*}
where $Z$ is the normalizing constant in $P(\cdot | c)$.

It should be noted that these metrics are only meaningful when compared with results on the same problem because an identical unknown constant $\log Z$ has been added to all of them.

For each question, we calculate the average KL divergence between all the samples and the language model.
We also present the average KL divergence of the four datasets, weighted by the number of problems in each dataset.

To see how each method converges to the oracle distribution, we also calculate the KL divergence between a sliding window of $n$ samples and the language model, where $n$ is one-fourth of the sampling size for each problem.

\paragraph{Baselines}
We compare with two baselines, constrained decoding (CD) and Adaptive Sampling with Approximate Expected Futures (ASAp) \cite{gad}.
ASAp can be regarded as a variant of \method{}, where after sampling a sequence from $P_Q$, instead of backtracking and adjusting all the previous sampling results, it just continues to use the current partially generated prefix.
It uses the information collected in $Q$ only in the later samples on the same problem.
To match its behavior, we also share the estimation data $Q$ between different generations; only the sample data $N$ is updated for a new sample.

\paragraph{Implementation}
We reuse ASAp code to replicate their settings and the same transformers-CFG\cite{cp} library for the constrainer.

\paragraph{Models}
Following previous work\cite{gad}, we use the Mistral-7B\cite{mistral7b} model for all decoding algorithms.

\input{floats/kls-table}

\paragraph{Results}
As shown in Table~\ref{tab:kl},
\textbf{in all four datasets, \method{} performs significantly better than both constrained decoding and ASAp.}
In both the SLIA and CP datasets, constrained decoding and ASAp perform roughly the same, while \method{} has a large advantage.
In the INV-BV dataset, all three methods perform roughly the same.
This might indicate that the portion of illegal tokens on different tokens in this dataset is roughly the same, so constrained decoding can perform roughly the same as \method{}.
The binary dataset is the only dataset where ASAp performs much better than constrained decoding, while \method{} still has some advantage over ASAp.

\input{floats/kls-figure}

Figure~\ref{fig:kls} shows how the KL divergence evolves in repeated sampling.
The point with the abscissa $x$ in the figure represents the KL divergence between the distribution formed by the $x$th to the $(n+x)$th samples and the model distribution.

\textbf{With more and more information stored in $Q$, decoding results of \method{} become better and better.}
On the SLIA, CP and binary dataset, \method{} has a low KL divergence from the start and continues to be low for the entire generation, indicating that \method{} has already found a good distribution from the first one-fourth samples.
On the INV-BV dataset, the KL divergence of \method{} declines greatly, indicating that \method{} is still learning new things in repeated sampling.

%% file: floats/tfv1-results.tex
\begin{table*}[t]
\centering
\caption{Exact match on \tfapis{} in different TensorFlow versions.}
\begin{tabular}{l|rrrrr|rrrrr}
\hline
& \multicolumn{5}{c|}{TensorFlow v1 setting} & \multicolumn{5}{c}{TensorFlow v2 setting}\\
Method & EM@1 & EM@3 & EM@5 & EM@10 & EM@20 & EM@1 & EM@3 & EM@5 & EM@10 & EM@20\\
\hline
\textit{Qwen2.5 Coder 7B} & & & & & & & \\
Unconstrained decoding & 61.58\% & 73.07\% & 81.51\% & 88.85\% & 92.84\% & 7.16\% & 19.04\% & 26.17\% & 37.77\% & 50.60\% \\
Constrained decoding & 66.11\% & 76.35\% & 83.50\% & 90.10\% & 94.03\% & 18.38\% & 31.84\% & 39.63\% & 50.34\% & 60.86\% \\
\method{} (Ours) & \textbf{80.19\%} & \textbf{87.36\%} & \textbf{91.86\%} & \textbf{95.96\%} & \textbf{98.33\%} & \textbf{54.89\%} & \textbf{62.51\%} & \textbf{68.87\%} & \textbf{76.52\%} & \textbf{82.82\%} \\
\hline
\textit{DeepSeek Coder Base 6.7B} & & & & & & & \\
Unconstrained decoding & 62.53\% & 68.63\% & 78.50\% & 88.43\% & 94.27\% & 5.97\% & 16.91\% & 24.12\% & 36.49\% & 50.60\% \\
Constrained decoding & 65.63\% & 73.34\% & 82.03\% & 90.26\% & 94.75\% & 10.98\% & 25.31\% & 33.32\% & 45.90\% & 59.43\% \\
\method{} (Ours) & \textbf{83.29\%} & \textbf{86.98\%} & \textbf{92.47\%} & \textbf{96.48\%} & \textbf{98.09\%} & \textbf{50.60\%} & \textbf{56.62\%} & \textbf{64.76\%} & \textbf{74.57\%} & \textbf{82.82\%} \\
\hline
\textit{StarCoder2 7B} & & & & & & & \\
Unconstrained decoding & 72.79\% & 76.37\% & 85.12\% & 92.30\% & 96.18\% & 1.67\% & 12.66\% & 18.62\% & 29.76\% & 44.15\% \\
Constrained decoding & 73.99\% & 79.99\% & 87.36\% & 93.23\% & 96.42\% & 10.26\% & 21.48\% & 28.83\% & 40.48\% & 53.70\% \\
\method{} (Ours) & \textbf{83.53\%} & \textbf{88.78\%} & \textbf{92.97\%} & \textbf{96.36\%} & \textbf{98.33\%} & \textbf{45.35\%} & \textbf{50.90\%} & \textbf{59.70\%} & \textbf{69.75\%} & \textbf{77.57\%} \\
\hline
\textit{CodeLlama Python 7B} & & & & & & & \\
Unconstrained decoding & 59.90\% & 70.63\% & 78.86\% & 86.61\% & 91.41\% & 11.93\% & 25.20\% & 32.52\% & 43.46\% & 55.37\% \\
Constrained decoding & 66.35\% & 76.36\% & 83.47\% & 89.78\% & 93.56\% & 21.24\% & 38.35\% & 46.12\% & 56.96\% & 67.54\% \\
\method{} (Ours) & \textbf{83.29\%} & \textbf{89.37\%} & \textbf{92.82\%} & \textbf{95.70\%} & \textbf{97.14\%} & \textbf{60.86\%} & \textbf{68.20\%} & \textbf{73.37\%} & \textbf{79.56\%} & \textbf{85.20\%} \\
\hline
\end{tabular}
\label{tab:tf1-apis}
\end{table*}

%% file: floats/tfv1-example.tex
\begin{figure}[h]
\centering

\begin{minted}[escapeinside=!!,mathescape=true,frame=lines]{python}
import tensorflow as tf

def mean_pairwise_squared_error(*args, **kwargs):
    """
    wrapper for losses.mean_pairwise_squared_error
    """
    return tf!$\big|$!
\end{minted}
  \caption{An example code in the \tfapis{} dataset.}
  \label{fig:tf-old:ex}
\end{figure}

%% file: floats/tfv1real-results.tex
\begin{table}[t]
\centering
\caption{Exact match on \tfreal{}.}
\begin{tabular}{l|rrrrr}
\hline
& \multicolumn{5}{c}{Exact Match @ $k$ ($\uparrow$)} \\
Method & EM@1 & EM@3 & EM@5 & EM@10 & EM@20\\
\hline
\textit{Qwen2.5C 7B} \\
Unconstrained & 21.70\% & 29.45\% & 35.28\% & 43.88\% & 53.00\% \\
Constrained & 35.60\% & 41.86\% & 47.64\% & 55.29\% & 62.30\% \\
\method{} & \textbf{49.30\%} & \textbf{53.87\%} & \textbf{59.28\%} & \textbf{65.84\%} & \textbf{71.70\%} \\
\hline
\textit{DSC Base 6.7B} \\
Unconstrained & 18.90\% & 25.62\% & 30.93\% & 38.66\% & 47.10\% \\
Constrained & 31.10\% & 38.45\% & 44.09\% & 51.23\% & 57.70\% \\
\method{} & \textbf{41.90\%} & \textbf{49.45\%} & \textbf{55.26\%} & \textbf{62.25\%} & \textbf{68.30\%} \\
\hline
\textit{SC2 7B} \\
Unconstrained & 15.90\% & 21.91\% & 27.33\% & 35.59\% & 44.70\% \\
Constrained & 29.80\% & 35.55\% & 41.69\% & 49.59\% & 56.60\% \\
\method{} & \textbf{41.40\%} & \textbf{47.60\%} & \textbf{53.85\%} & \textbf{61.62\%} & \textbf{68.10\%} \\
\hline
\textit{CL Python 7B} \\
Unconstrained & 19.60\% & 26.88\% & 32.26\% & 40.48\% & 49.40\% \\
Constrained & 34.40\% & 40.77\% & 46.99\% & 55.10\% & 62.30\% \\
\method{} & \textbf{46.80\%} & \textbf{52.61\%} & \textbf{58.75\%} & \textbf{66.41\%} & \textbf{73.00\%} \\
\hline
\end{tabular}
\label{tab:tf1-apis-real}
\end{table}

%% file: floats/codegen-results.tex
\begin{table*}[t]
\centering
\caption{Pass@$k$ on general-purpose code generation benchmarks.}
\begin{tabular}{l|rrrrr|rrrrr}
\hline
& \multicolumn{5}{c|}{HumanEval} & \multicolumn{5}{c}{MBPP}\\
Method & pass@1 & pass@3 & pass@5 & pass@10 & pass@20 & pass@1 & pass@3 & pass@5 & pass@10 & pass@20\\
\hline
\textit{Qwen2.5 Coder 7B} & & & & & & & \\
Unconstrained decoding & 65.41\% & 63.05\% & 73.66\% & 84.36\% & 90.57\% & 66.92\% & 72.28\% & 79.10\% & 85.04\% & 88.72\% \\
Constrained decoding & 64.15\% & 70.66\% & 79.00\% & 86.08\% & 89.31\% & 66.67\% & 75.13\% & 80.60\% & 85.56\% & 88.46\% \\
\method{} (Ours) & \textbf{69.18\%} & \textbf{72.22\%} & \textbf{79.98\%} & \textbf{86.65\%} & \textbf{92.45\%} & \textbf{70.77\%} & \textbf{76.26\%} & \textbf{81.55\%} & \textbf{86.59\%} & \textbf{89.74\%} \\
\hline
\textit{DeepSeek Coder Base 6.7B} & & & & & & & \\
Unconstrained decoding & 44.65\% & 46.52\% & 56.26\% & 68.23\% & 77.36\% & 36.41\% & 65.62\% & 73.77\% & 81.17\% & 85.38\% \\
Constrained decoding & 47.17\% & 51.88\% & 60.36\% & 70.00\% & 78.62\% & 60.00\% & 67.87\% & 74.26\% & 80.24\% & 84.10\% \\
\method{} (Ours) & \textbf{49.69\%} & \textbf{55.43\%} & \textbf{63.84\%} & \textbf{73.16\%} & \textbf{80.50\%} &  \textbf{63.85\%} & \textbf{70.29\%} & \textbf{76.54\%} & \textbf{82.26\%} & \textbf{85.90\%} \\
\hline
\textit{StarCoder2 7B} & & & & & & & \\
Unconstrained decoding & \textbf{37.74\%} & 28.00\% & 35.79\% & 46.85\% & 57.86\% & \textbf{56.67\%} & 45.94\% & 55.72\% & 66.36\% & 74.36\% \\
Constrained decoding & 37.11\% & 37.67\% & 46.35\% & 58.18\% & 68.55\% & 52.05\% & 57.43\% & 65.09\% & 73.17\% & 78.72\% \\
\method{} (Ours) & \textbf{37.74\%} & \textbf{40.72\%} & \textbf{49.62\%} & \textbf{60.75\%} & \textbf{69.18\%} & 54.36\% & \textbf{58.60\%} & \textbf{66.17\%} & \textbf{74.33\%} & \textbf{80.51\%} \\
\hline
\textit{CodeLlama 7B} & & & & & & & \\
Unconstrained decoding & 33.96\% & 32.15\% & 39.50\% & 49.89\% & 61.64\% & \textbf{48.72\%} & 45.68\% & 54.57\% & 64.42\% & 71.54\% \\
Constrained decoding & 34.59\% & 35.19\% & 42.62\% & 53.43\% & \textbf{65.41\%} & 45.90\% & 52.25\% & 60.19\% & 68.89\% & 75.64\% \\
\method{} (Ours) & \textbf{36.48\%} & \textbf{35.89\%} & \textbf{43.41\%} & \textbf{53.89\%} & 64.15\% & 47.69\% & \textbf{53.26\%} & \textbf{61.25\%} & \textbf{70.34\%} & \textbf{77.18\%} \\
\hline
\end{tabular}
\label{tab:codegen}
\end{table*}

%% file: floats/kls-table.tex
\begin{table}[h]
\centering
\caption{Comparison of KL divergence of different methods on DSL datasets. CD means constrained decoding.}
\begin{tabular}{l|rrrr|r}
\hline
& \multicolumn{5}{c}{KL-divergence ($\downarrow$)}\\
Method & SLIA & INV-BV & CP & binary & average\\
\hline
CD & 11.37 & 7.13 & 19.86 & 23.05 & 11.46 \\
ASAp & 12.05 & 7.76 & 19.85 & 19.87 & 11.90 \\
\method{} & \textbf{9.31} & \textbf{7.04} & \textbf{16.86} & \textbf{19.53} & \textbf{9.97} \\
\hline
\end{tabular}
\label{tab:kl}
\end{table}

%% file: floats/kls-figure.tex
\begin{figure}[h]
\centering
\begin{subfigure}{.24\linewidth}
  \centering
  \includegraphics[width=\linewidth]{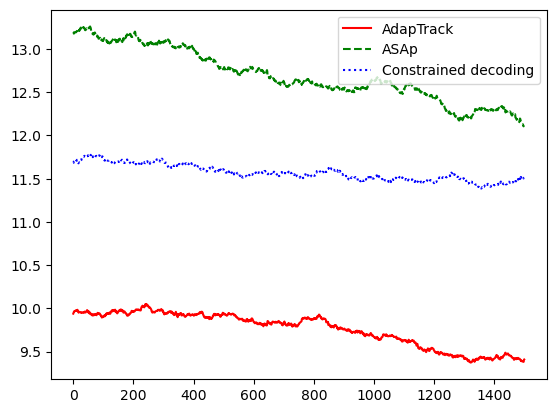}
  \caption{SLIA}
  \label{fig:slia}
\end{subfigure}\hfill
\begin{subfigure}{.24\linewidth}
  \centering
  \includegraphics[width=\linewidth]{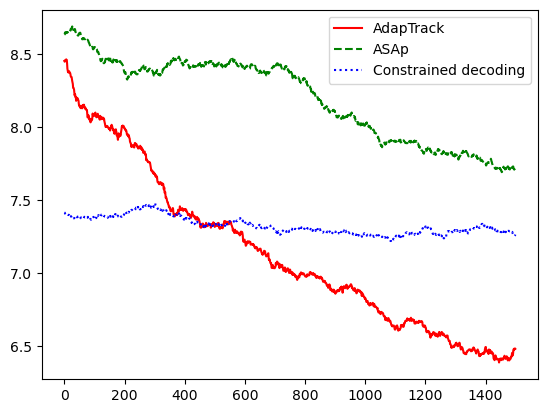}
  \caption{INV-BV}
  \label{fig:bv4}
\end{subfigure}\hfill
\begin{subfigure}{.24\linewidth}
  \centering
  \includegraphics[width=\linewidth]{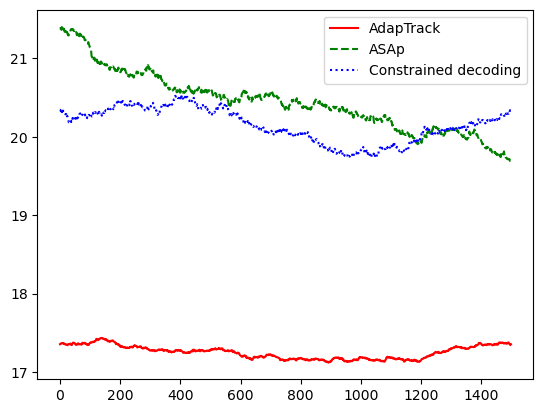}
  \caption{CP}
  \label{fig:cp}
\end{subfigure}\hfill
\begin{subfigure}{.24\linewidth}
  \centering
  \includegraphics[width=\linewidth]{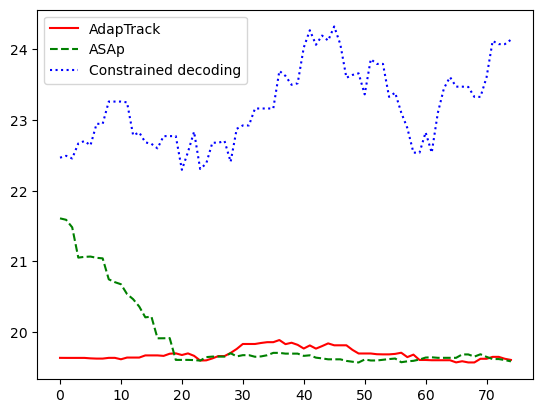}
  \caption{Binary}
  \label{fig:binary}
\end{subfigure}
\caption{The change of KL divergence ($\downarrow$) between every 500 samples and the oracle distribution on different datasets.}
\label{fig:kls}
\end{figure}

%% file: chapters2/discussion.tex
\section{Discussion}
\label{sec:discuss}

\subsection{Robustness analysis}

\subsubsection{Temperature}
Performance of decoding method often depends on the sampling temperature.
In previous RQs, the temperature is always 1.
Therefore, we conduct experiments to see how the performance of \method{} changes with different sampling temperatures, i.e. 0.1, 0.2, 0.3, 0.4, 0.5, 0.6, 0.7, 0.8, 0.9.
The experiments are conducted with best-performing models on \tfapis{} and \tfreal{}, i.e. CodeLlama Python 7B and Qwen2.5 Coder 7B, respecitively.

\input{floats/ablation-temperature}

As shown in figure~\ref{fig:temperature}, under temperatures from 0.0 to 1.0, for EM@$k$ where $k \in \{1, 3, 5, 10, 20\}$, \method{} always outperforms both constrained and unconstrained decoding.
As the temperature increases, the performance of all methods increases, because appropriately higher temperature encourages the model to explore answers in a larger range, and EM@$k$ is 1 if any of the $k$ attempts hits the correct answer, except for EM@1, where greedy sampling always gives the same result.

\subsubsection{Model size}
Performance of decoding method often depends on the model size.
Some methods only work with small models or large models.
In previous RQs, we always experiment with models sized roughly 7B.
Therefore, we conduct experiments to see how the performance of \method{} changes with differently sized models.
The experiments are conducted with Qwen2.5 Coder series (0.5B, 1.5B, 3B, 7B, 14B, 32B) on \tfapis{} and \tfreal{} datasets, because Qwen2.5 Coder series provide a balanced and various range of choices for model sizes.
To reduce needed GPU memory, for models larger than 7B, we use the \texttt{bfloat16} format during inference.

\input{floats/ablation-size}

As shown in figure~\ref{fig:size}, under all different model sizes, on both \tfapis{} and \tfreal{} datasets, \method{} always outperforms both constrained decoding and unconstrained decoding.
As the model size increases, the performance of all methods on \tfapis{} is saturated or even decreases after a certain threshold.
One possible reason is that \tfapis{} is synthetic and too simple, so the difference in model size does not influence the result.
As the model size increases, all methods' performance on \tfreal{} increases, indicating larger models' better understanding in API completion.

\subsection{Efficiency and scalability of \method{}}
One might concern the efficiency and scalability of \method{} from two perspectives: (1) updating $Q$ and $N$ from step~\ref{step:updateQ} to step~\ref{step:rej:end} in algorithm~\ref{alg:mine} can be inefficient; (2) \method{} may need too much time before generating a complete sequence.

\subsubsection{Parallel update}
The update of $Q$ and $N$ can be parallelized on GPU, rather than performed token by token.
If $Q[x]$ is updated from $1$ to $q$, where $x$ becomes an old prefix during the update, if $x = \textsc{Concat}(s, x')$, following formula can be used to update $Q[s]$,
$$
Q[s] \gets Q[s] - \Pof{LM}(x' | s) (1 - q),
$$
where $\Pof{LM}(x' | s)$ can be computed in parallel for all prefix $s$ of $x$.
It is then natural to compute the updated $P_Q(\cdot | s)$ in parallel, and perform the rejection sampling needed.
It should be noted that the rejection sampling and the sampling after rejection can also be fused into one sampling, further reducing branching in the calculation.

\subsubsection{Number of LM calls}
The most computationally intensive part of the decoding process is the invocation of the language model.
We count the number of calls to the language model per sample in the RQ1 experiments, as shown in Table~\ref{tab:average-length}.
In the TensorFlow v1 setting, \method{} invokes the language model on average 20\% more frequently than constrained decoding, whereas, in the TensorFlow v2 setting, this number rises to an average of 60\% more calls than constrained decoding.
Similar numbers also apply to the averaged maximum number of calls to LM.
Considering that EM@1 of \method{} is often comparable to EM@5 of constrained decoding, this level of increase in the number of invocations can be justified.

\input{floats/average-length}

If theoretical completeness is sacrificed, we can take additional constraints on \method{} to avoid extremely long backtracking distances.
For example, a lower bound can be set for the backtracking probability, prohibiting backtracking if the probability falls below this threshold.
An upper limit can also be set for the backtracking distance, disallowing backtracking beyond this specified maximum.

\input{floats/ablation-lengths}

In particular, as shown in Figure~\ref{fig:exp:lengths}, we conduct experiments with limited backtracking distances (including 1, 2, 4, 8) with CodeLlama Python 7B on \tfapis{} and Qwen2.5 Coder 7B on \tfreal{}, respectively.
For both models and datasets, the performance increases as the allowed backtracking distance increases, but stops increasing after a certain threshold.
From the figure, we can see that a backtracking distance of 2 or 4 is roughly enough for \tfapis{} (v1 settings), 4 or 8 for \tfapis{} (v2 settings), and 2 for \tfreal{}.
This indicates that a short backtracking distance is enough for most of \method{}'s performance.

\subsection{Threats to validity}

\subsubsection{Dataset selection}
The first threat lies in the selection of dataset: Does this problem of constrained decoding happen only in deprecated APIs or only in API completion?
We mitigate this threat by experiments in RQ1 where we ask the model to generate a TensorFlow v1 API in a TensorFlow v1 setting, where these APIs are not considered deprecated, which can be fixed by \method{} to provide significant improvement.
Experiments on general code generation benchmarks in RQ3 show that such distortion also happens in general code generation, and \method{} can provide performance improvement by fixing such distortion.

\subsubsection{Model selection}
The second threat lies in the selection of the models: Is this issue due to the limitation in model’s capability?
We mitigate this threat by using multiple models in experiments, and using differently-sized models in robustness analysis.
Although they have different capabilities, \method{} shows a consistent improvement.
It should be noted that this is an inherent problem of constrained decoding.
As long as constrained decoding is employed, the problem will persist to some extent.

%% file: floats/ablation-temperature.tex
\begin{figure}[h]
\centering
\begin{subfigure}{.33\linewidth}
  \centering
  \includegraphics[width=\linewidth]{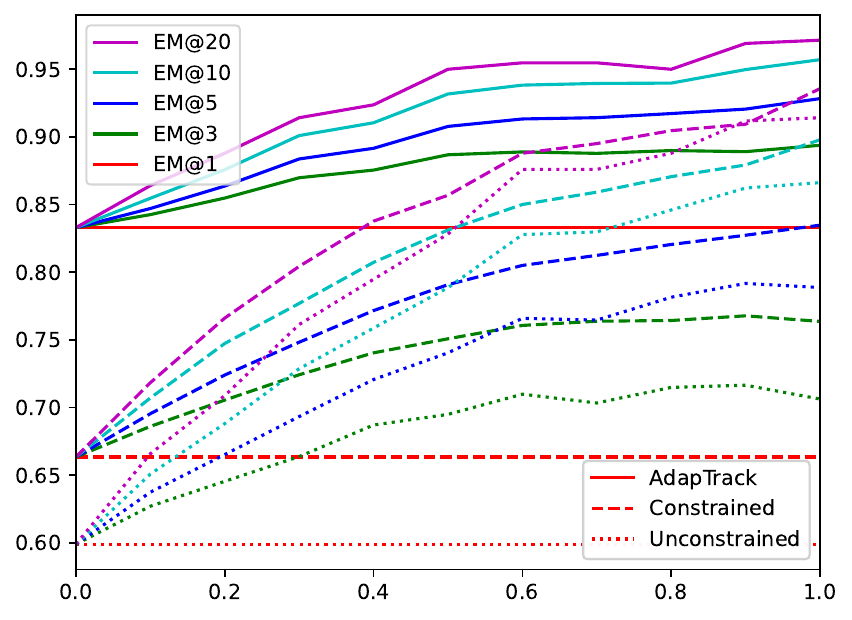}
  \caption{\tfapis{} (v1 setting)}
  \label{fig:temp:tfv1-in-v1}
\end{subfigure}\hfill
\begin{subfigure}{.33\linewidth}
  \centering
  \includegraphics[width=\linewidth]{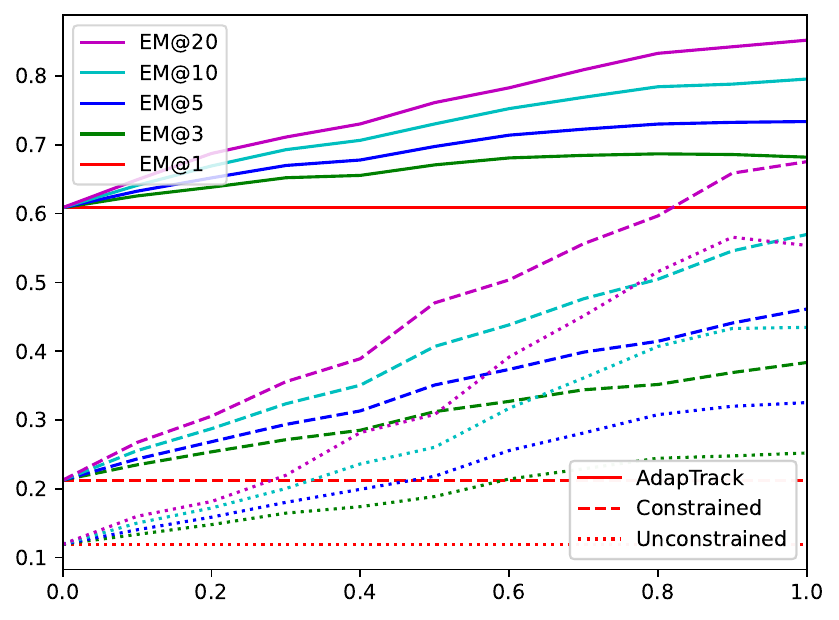}
  \caption{\tfapis{} (v2 setting)}
  \label{fig:temp:tfv1-in-v2}
\end{subfigure}\hfill
\begin{subfigure}{.33\linewidth}
  \centering
  \includegraphics[width=\linewidth]{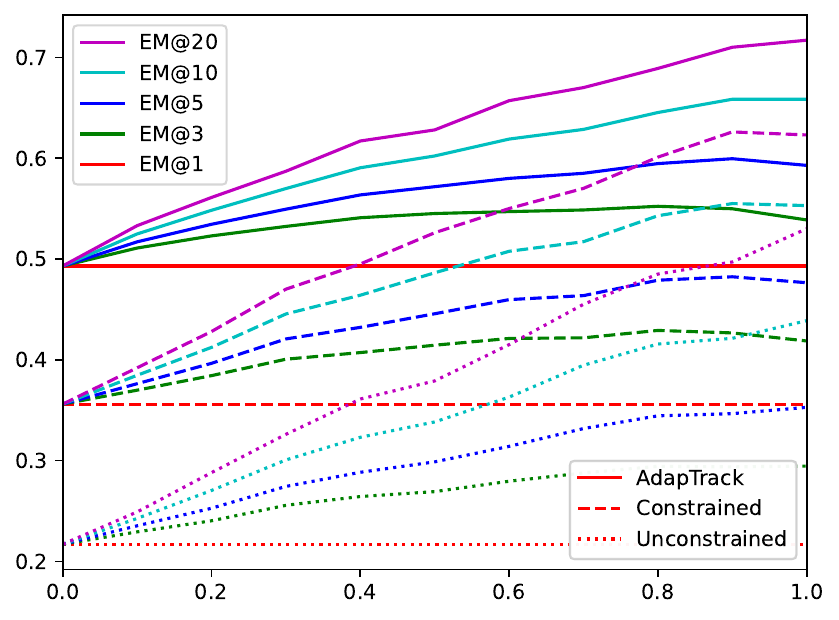}
  \caption{\tfreal{}}
  \label{fig:temp:tfv1real}
\end{subfigure}
\caption{The change of EM@$k$ in different temperature.}
\label{fig:temperature}
\end{figure}

%% file: floats/ablation-size.tex
\begin{figure}[h]
\centering
\begin{subfigure}{.33\linewidth}
  \centering
  \includegraphics[width=\linewidth]{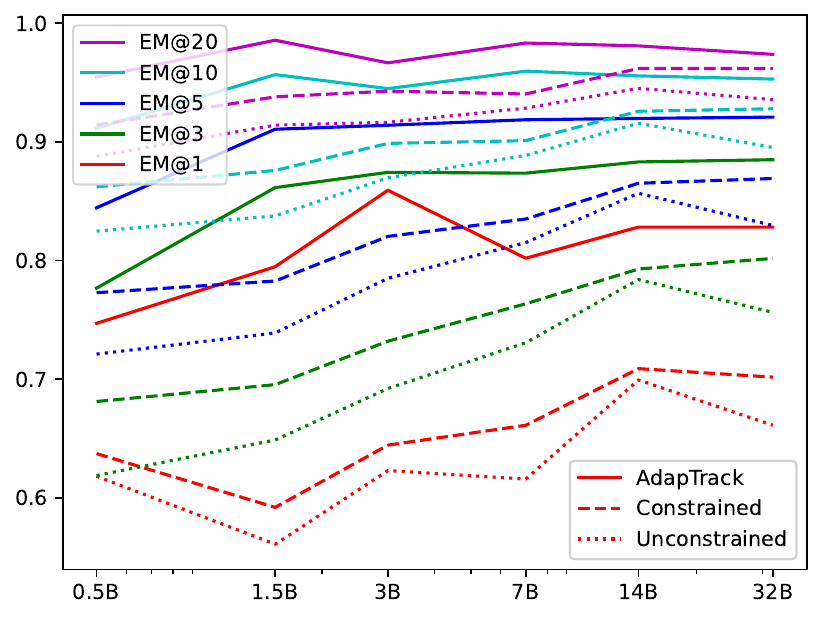}
  \caption{\tfapis{} (v1 setting)}
  \label{fig:size:tfv1-in-v1}
\end{subfigure}\hfill
\begin{subfigure}{.33\linewidth}
  \centering
  \includegraphics[width=\linewidth]{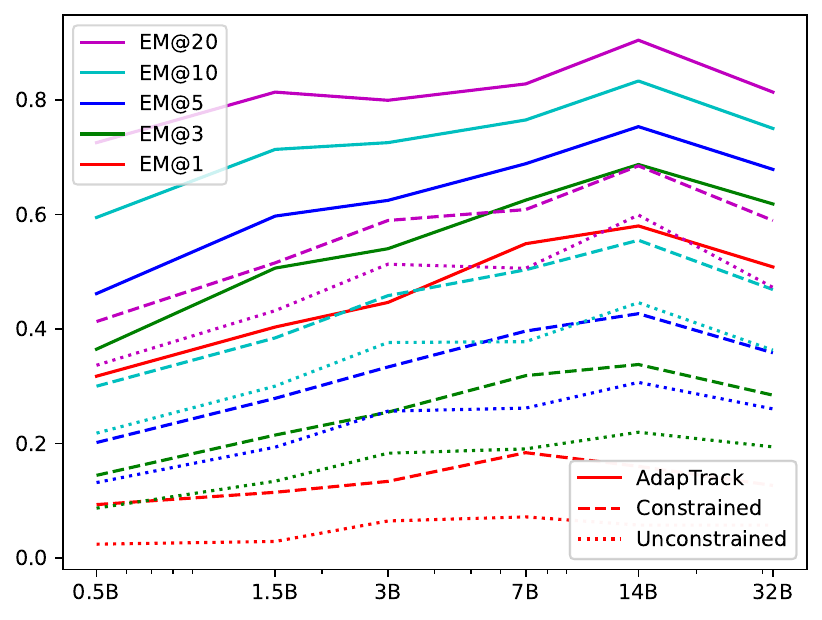}
  \caption{\tfapis{} (v2 setting)}
  \label{fig:size:tfv1-in-v2}
\end{subfigure}\hfill
\begin{subfigure}{.33\linewidth}
  \centering
  \includegraphics[width=\linewidth]{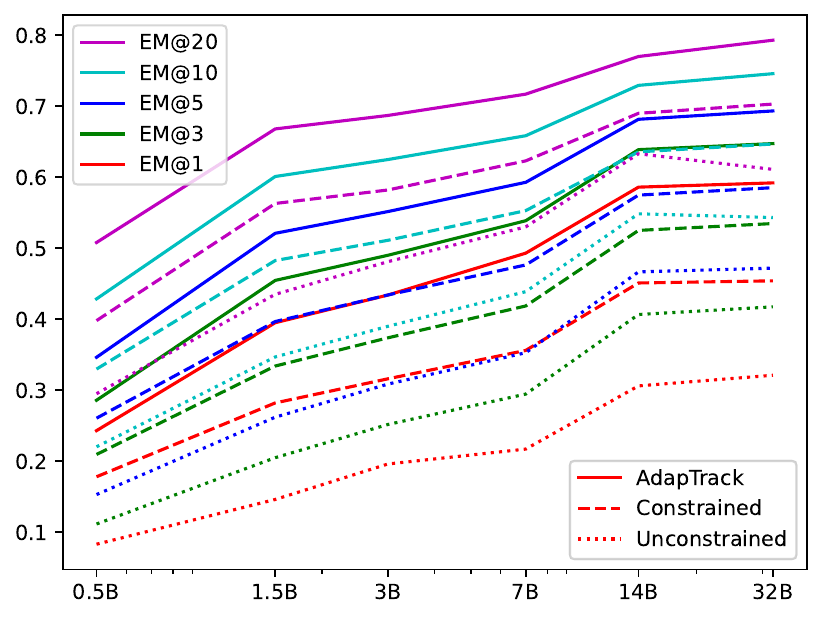}
  \caption{\tfreal{}}
  \label{fig:size:tfv1real}
\end{subfigure}
\caption{The change of EM@$k$ for differently-sized models.}
\label{fig:size}
\end{figure}

%% file: floats/average-length.tex
\begin{table}[h]
\centering
\caption{Numbers of calls to language model in RQ1.}
\begin{tabular}{l|r|r}
\hline
& \multicolumn{2}{c}{calls to LM ($\downarrow$)} \\
& \tfapis{} (v1 settings) & \tfapis{} (v2 settings) \\
Method & avg (min, max) & avg (min, max) \\
\hline
\textit{Qwen2.5 Coder 7B} & & \\
Unconstrained decoding & 5.85 (4.12, 8.47) & 5.85 (4.12, 8.47) \\
Constrained decoding & 4.64 (2.75, 6.07) & 5.04 (2.91, 7.82) \\
\method{} (Ours) & 5.56 (4.12, 7.42) & 7.94 (4.41, 11.73) \\
\hline
\textit{DeepSeek Coder Base 6.7B} & & \\
Unconstrained decoding & 10.36 (6.32, 15.59) & 10.36 (6.32, 15.59) \\
Constrained decoding & 8.27 (4.93, 11.43) & 8.82 (4.74, 14.40) \\
\method{} (Ours) & 9.85 (6.24, 11.37) & 13.90 (6.70, 22.87) \\
\hline
\textit{StarCoder2 7B} & & \\
Unconstrained decoding & 9.11 (5.75, 13.65) & 9.11 (5.75, 13.65) \\
Constrained decoding & 7.62 (4.47, 10.03) & 7.70 (4.33, 12.50) \\
\method{} (Ours) & 8.83 (5.93, 12.60) & 12.08 (5.47, 20.32) \\
\hline
\textit{CodeLlama Python 7B} & & \\
Unconstrained decoding & 10.68 (7.56, 12.90) & 10.68 (7.56, 12.90) \\
Constrained decoding & 8.41 (5.38, 10.75) & 9.73 (5.88, 14.02) \\
\method{} (Ours) & 10.03 (7.36, 13.45) & 14.29 (8.81, 20.25) \\
\hline
\end{tabular}
\label{tab:average-length}
\end{table}

%% file: floats/ablation-lengths.tex
\begin{figure}[h]
\centering
\begin{subfigure}{.33\linewidth}
  \centering
  \includegraphics[width=\linewidth]{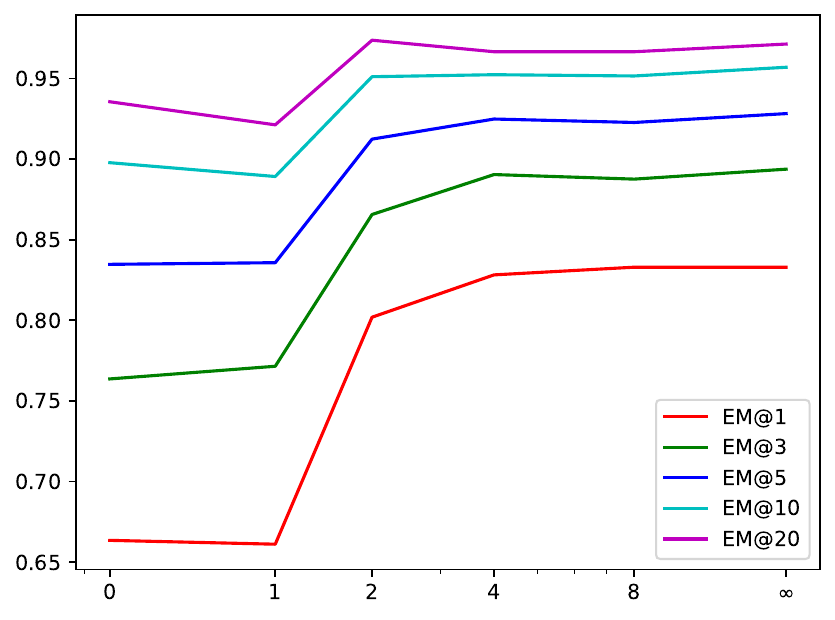}
  \caption{\tfapis{} (v1 setting)}
  \label{fig:length:tfv1-in-v1}
\end{subfigure}\hfill
\begin{subfigure}{.33\linewidth}
  \centering
  \includegraphics[width=\linewidth]{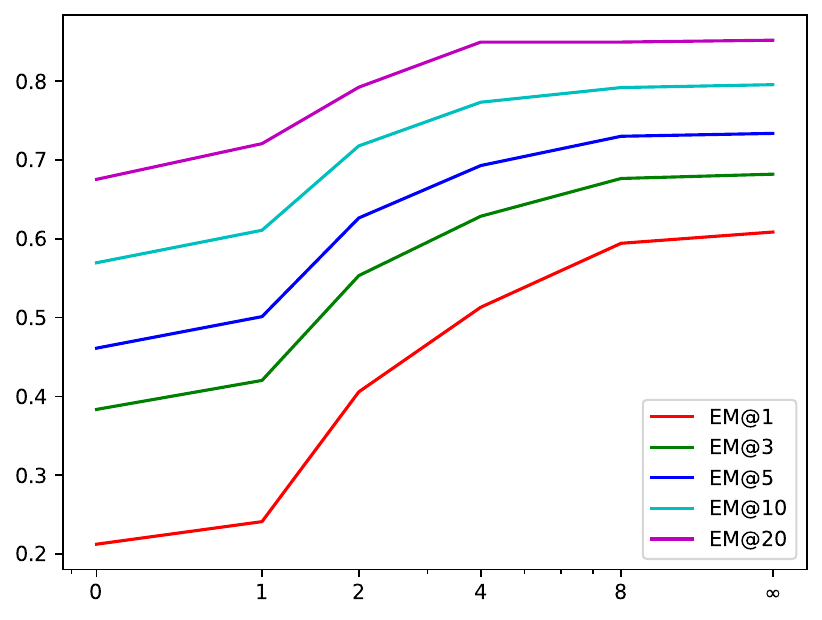}
  \caption{\tfapis{} (v2 setting)}
  \label{fig:length:tfv1-in-v2}
\end{subfigure}\hfill
\begin{subfigure}{.33\linewidth}
  \centering
  \includegraphics[width=\linewidth]{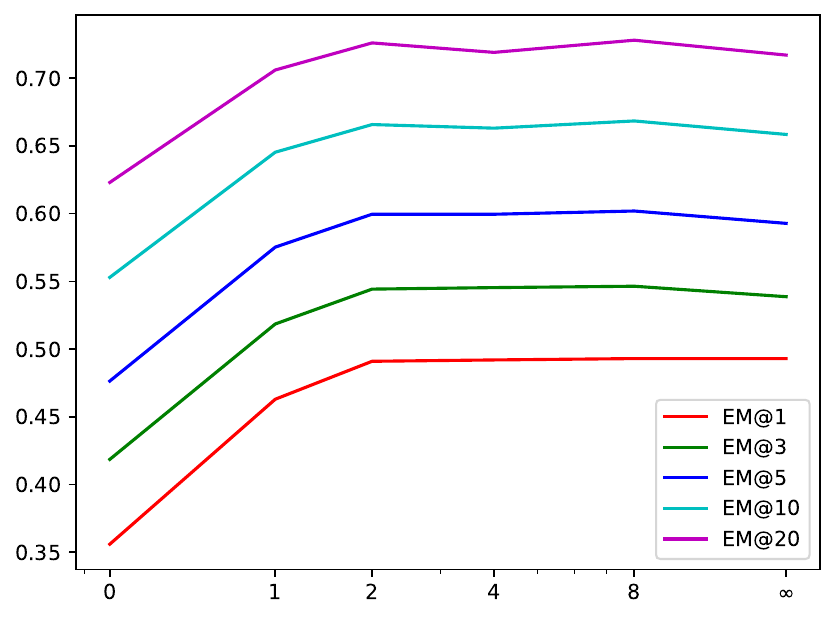}
  \caption{\tfreal{}}
  \label{fig:length:tfv1real}
\end{subfigure}
\caption{The change of EM@$k$ with different maximum backtrack distance.}
\label{fig:exp:lengths}
\end{figure}

%% file: chapters2/related.tex
\section{Related work}
\label{sec:related}

\subsection{API-aware code generation}
APIs evolve over time, and code LLMs need to handle that.
CERT\cite{cert} introduces a sketcher to predict the sketch of API usage, and then uses a generator to use the API in the relevant context.
Toolcoder\cite{toolcoder} introduces an API search tool to help the model search for wanted APIs.
\citeauthor{llm-deprecated-api} \cite{llm-deprecated-api} conduct the first evaluation study on the usage of deprecated APIs in LLM-based code completion.
VersiCode\cite{versicode} dataset focuses on version-specific code completion and version-aware code migration.
CodeUpdateArena\cite{codeupdatearena} dataset constructs synthetic API function
update to evaluate LLM's ability to handle API update.

\subsection{Constrained decoding in code generation}
Constrained decoding is a greedy method to ensure the generated sequence statisfy certain constraints.
Picard \cite{picard} and Synchromesh \cite{synchromesh} introduce incremental parser into code generation to help the model reject inadmissible tokens at each decoding step.
Later grammar-based constrainers \cite{koo2024automatabased,domino,xgrammar,greatgramma} focus on the acceleration of constraint calculation.
Monitor-guided decoding \cite{mgd} and Repilot \cite{repilot} introduce constrained decoding in the generation of general-purpose programming languages, by communicating with static analysis tools in IDEs, e.g., language servers in Visual Studio Code, to collect information about the whole repository to help API generation.
\citeauthor{type-constraint} \cite{type-constraint} introduce type system into constrained decoding, ensuring the generated program to be type-safe.

\subsection{Fixing greedy constrained decoding}
Na\"ive greedy constrained decoding is known to be problematic, and there are some work trying to fix that.
GeDi\cite{gedi} introduces another trained model to model the expected validity of each prefix, but training another model is a large overhead.
ASAp\cite{gad} samples multiple sequences from the language model and use the validity information from previous samples to direct future sampling.
This can be regarded as a variant of \method{} without backtracking.
Gen-C\cite{local-constrain} models the constraint with constraint circuits and resamples using the pseudolikelihood in the neighborhood of a sentence, but constraint circuits are only constructible for regular expressions and, therefore, unsuitable for code generation.

\subsection{Non-linear code generation techniques}
Non-linear code generation techniques is not unheard of.
PG-TD \cite{pg-td} uses Monte-Carlo tree search, and executes public test cases to determine the quality of generated code.
Rocode \cite{rocode} detects errors after generating every statement by running static analysis tools and executing public test cases, and rollbacks to the statement of the error report or earlier statements with high generation uncertainty.
However, these methods are not designed to address distribution distortion and consequently still produce distorted distributions, while \method{} specifically targets this issue and provides provably undistorted results.
These methods often use public test cases to filter out incorrect programs, but, while possible, similar procedure is not executed in \method{}, because \method{} focuses more on the distribution distortion problem.

%% file: chapters2/conclusion.tex
\section{Conclusion and future work}
\label{sec:conclude}

In this work, we 
propose a new decoding method, \method{}, to fix the distribution distortion problem in constrained decoding by introducing backtracking that is adaptive to the invalidated probability mass into the generation process.
We collect two datasets, the synthetic \tfapis{} dataset and the real-world \tfreal{} dataset to verify the effectiveness of \method{}.
Experiments on general code generation benchmarks verifies the generalizability of \method{}.
We give theoretical proof of our method that it does follow the output intent of the model and conduct experiments on DSL
problems to verify that.

In the future, we will explore this phenomenon on more APIs, and explore how to help the model align better with the constraints.